%% file: BI_QI.tex
\documentclass[aps,10pt,showpacs,pra,amsmath,amssymb,floatfix,superscriptaddress,longbibliography,twocolumn]{revtex4-2}
\usepackage{xcolor}
\usepackage{cancel}

\input{preamble}

\usepackage{hyperref}
 \usepackage{graphicx}
\usepackage{mathrsfs}

\begin{document}
	
	\title{On the  Bargmann invariants for quantum imaginarity}
	
	\author{Mao-Sheng Li}
	\email{li.maosheng.math@gmail.com}
	\affiliation{School  of Mathematics, South China University of Technology, GuangZhou 510640, China}
	
	\author{Yi-Xi Tan}
	\affiliation{School  of Mathematics, South China University of Technology, GuangZhou 510640, China}
	
	\date{\today}
	
	\begin{abstract}
	The imaginary in quantum theory plays a crucial role in describing quantum coherence and is widely applied in quantum information tasks such as state discrimination, pseudorandomness generation, and quantum metrology. A recent paper by Fernandes et al. [C. Fernandes, R. Wagner, L. Novo, and E. F. Galvão, \href{https://doi.org/10.1103/PhysRevLett.133.190201}{Phys. Rev. Lett. 133, 190201 (2024)}] showed how to use the Bargmann invariant to witness the imaginarity of a set of quantum states. In this work, we delve into the structure of Bargmann invariants and their quantum realization in qubit systems. First, we present a characterization of special sets of Bargmann invariants (also studied by Fernandes et al. for a set of four states) for a general set of 
	$n$ quantum states. Then, we study the properties of the relevant Bargmann invariant set 
	$\mathcal{B}_n$	and its quantum realization in qubit systems. Our results provide new insights into the structure of Bargmann invariants, contributing to the advancement of quantum information techniques, particularly within qubit systems.
	
	\end{abstract}
	
	\maketitle
	
	 \section{Introduction}
	 Quantum theory's reliance on complex numbers is well-established~\cite{Griffiths:2004, GIBBONS1993117}. Fundamental concepts in quantum theory, such as quantum states, measurements, and evolutions, are intrinsically tied to complex numbers. The quantification of the imaginary parts of quantum states, referred to as quantum imaginarity, has emerged as a valuable resource, driving significant advancements in quantum information science~\cite{RevModPhys.91.025001, baumgratz2014quantifying, RevModPhys.89.041003, dressel2012significance, hickey2018quantifying, kunjwal2019anomalous}. Although some works have explored the possibility of formulating quantum theory within a real vector space, proposing frameworks for quantum computing and information processing based on real-number operations~\cite{wootters2010entanglement, PhysRevA.87.052106, McKague2009simulating, Nielsen_Chuang_2010}, both theoretical and experimental results~\cite{Renou2021, PhysRevLett.126.090401, Miller_Imaginarity2022, PRL, Wu2024, Yao_Rule_out_Real_2024} have  established the necessity of quantum imaginarity for accurately modeling certain quantum phenomena. These findings  highlight the fundamental and indispensable role of quantum imaginarity in the precise description and manipulation of quantum systems. As a resource, quantum imaginarity has found applications in a wide range of quantum information tasks, including state discrimination~\cite{PhysRevA.104.042606}, pseudorandomness generation~\cite{PhysRevLett.114.160502}, and quantum metrology~\cite{PhysRevLett.96.010401}.

	  Since the seminal works~\cite{Renou2021, PhysRevLett.126.090401,Wu_PRA_2021}, numerous studies have focused on quantifying quantum imaginarity or the resource theory of quantum imaginarity, employing tools such as the \( l_1 \)-norm, trace norm, and various entropies, while also exploring its diverse applications~\cite{PRL, XueLi_Resource_Imaginarity2021, LiLuoSun2022, Chen2022MeasuresOI, WOS:001354223400001, Qi_Imaginarity_Super_2023, zhu2021hidingmasking, Wu2024, ChenPLA2024, Fan_KDI_2024, WeiFei_PRA_2024, ZhangLi_ImaginarityBroadcast2024, ChenHuangFei2024, WOS:001354223400001, PhysRevA.108.062203, du2024quantifyingimaginaritytermspurestate, Zhang2024Broad}. Similar to quantum coherence, the imaginarity of a state depends on the choice of the computational basis. By studying the imaginarity of a set of quantum states, one can achieve a basis-independent characterization of quantum imaginarity~\cite{oszmaniec2021measuring}.	  
	  More recently, in Ref.~\cite{PhysRevLett.133.190201}, the authors explored how to witness quantum imaginarity by examining unitary-invariants, specifically the Bargmann invariants~\cite{simon1993bargmann}, in sets of quantum states. They provided a comprehensive characterization of the Bargmann invariants for three pure states (i.e., \( \mathcal{B}_3 \)) and a partial characterization for four states, focusing on a subset \( \mathcal{B}_{4|\text{circ}} \) of the full set $\mathcal{B}_4$ of Bargmann invariants for four pure states. Most notably, they demonstrated that the imaginarity of four-state sets can be effectively witnessed through pairwise overlaps, a result that, intriguingly, does not extend to sets of three states. This pivotal finding significantly advances our understanding of quantum imaginarity but raises several intriguing questions, such as how to characterize the Bargmann invariants for sets with more than four states and how to realize these invariants in a qubit system.	  
	  In this work, we provide a deeper exploration of these questions and offer partial solutions to these challenges.
	  
	  The general content and structure of this paper are as follows:
	   In Sec.~\ref{sect:Pre}, we introduce the basic concepts of basis-independent imaginary sets, Bargmann invariants, and the relationship between them.
	  In Sec.~\ref{sect:Main}, we present a complete characterization for a subset, \( \mathcal{B}_{n \mid \text{circ}} \), of Bargmann invariants for \( n \) states and discuss some properties of the total Bargmann invariant set \( \mathcal{B}_n \).
	  In Sec.~\ref{sect:quantum realization}, we demonstrate that all Bargmann invariants in \( \mathcal{B}_{n \mid \text{circ}} \) can be realized using qubits.
	  Finally, we conclude and make a discussion  of our findings in Sec.~\ref{sect:con}.
	  
	\section{Preliminaries }\label{sect:Pre}
	
	Throughout this paper, $\mathbb{Z}, \mathbb{R}, \mathbb{C}$ denote the set of all integers, real numbers, and complex numbers, respectively. Let $\mathscr{H}$ be a quantum system of dimension $d$, $\mathbf{D}_d$ denote the set of all density matrices (self-adjoint, positive semidefinite matrices with trace 1), and $\mathbf{P}_d = \{\vert \psi \rangle \in \mathscr{H} \mid \langle \psi | \psi \rangle = 1 \}$ represent the set of pure states in the system $\mathscr{H}$. For simplicity, we will denote $\psi$ as the density matrix of the pure state $|\psi \rangle$, i.e., $\psi = |\psi \rangle \langle \psi |$.

	Given an ordered set of quantum states \( \vec{\rho} = (\rho_1, \dots, \rho_n) \in \mathbf{D}_d^n \), we are interested in determining whether there exists a basis such that all elements of \( \vec{\rho} \) have real entries. Equivalently, we seek to determine whether there exists a unitary \( U \) such that for all \( i \in \{1, 2, \dots, n\} \), the transformed states \( U \rho_i U^\dagger \) lie in \( \mathrm{Mat}_{d \times d}(\mathbb{R}) \).	
	If no such basis exists, we call the set \( \vec{\rho} \) a basis-independent imaginary set. The Bargmann invariants, defined as the quantity \( \mathrm{Tr}[\rho_1 \rho_2 \cdots \rho_n] \), are clearly basis-independent and are used to characterize such a set. In fact, the imaginarity of the value \( \mathrm{Tr}[\rho_1 \rho_2 \cdots \rho_n] \) implies the basis-independent imaginarity of the set \( \vec{\rho} = (\rho_1, \dots, \rho_n) \).

Let \( B_{n,d} \) denote the set of all Bargmann invariants of length \( n \) states in \( \mathbf{D}_d \), i.e.,
\[
B_{n,d} := \left\{ \mathrm{Tr}[\rho_1 \rho_2 \cdots \rho_n] \mid \rho_i \in \mathbf{D}_d, \forall i = 1, \dots, n \right\}.
\]
We are, in fact, more interested in considering pure states. Therefore, we define
\[
\mathcal{B}_{n,d} := \left\{ \mathrm{Tr}[\psi_1 \psi_2 \cdots \psi_n] \mid |\psi_i\rangle \in \mathbf{P}_d, \forall i = 1, \dots, n \right\}.
\]
  In this case, we have the identity
\[
\mathrm{Tr}[\psi_1 \psi_2 \cdots \psi_n] = \langle \psi_1 | \psi_2 \rangle \langle \psi_2 | \psi_3 \rangle \cdots \langle \psi_n | \psi_1 \rangle.
\]

This leads to the following ascending chain of sets:
\[
\mathcal{B}_{n,2} \subseteq \mathcal{B}_{n,3} \subseteq \cdots \subseteq \mathcal{B}_{n,d} \subseteq \mathcal{B}_{n,d+1} \subseteq \cdots
\]
and we define
\[
\mathcal{B}_n = \bigcup_{d=2}^{\infty} \mathcal{B}_{n,d}.
\]

Let \( \Psi = (|\psi_1 \rangle, \dots, |\psi_n \rangle), \Phi = (|\phi_1 \rangle, \dots, |\phi_n \rangle) \in \mathbf{P}_d^n \). We define \( \Psi \) to be unitary equivalent to \( \Phi \) (denoted \( \Psi \equiv \Phi \)) if there exists a unitary matrix \( U \) such that
\[
|\phi_j \rangle = U |\psi_j \rangle, \quad \forall j = 1, 2, \dots, n.
\]
The equivalence class of elements in \( \mathbf{P}_d^n \) under unitary transformations can be characterized by their Gram matrices. Specifically, we define \( G_\Psi \) as the \( n \times n \) matrix whose \( (i,j) \)-th entry is \( \langle \psi_i | \psi_j \rangle \).
Under this definition, two tuples \( \Psi, \Phi \in \mathbf{P}_d^n \) are unitarily equivalent if and only if their associated Gram matrices satisfy \( G_\Psi = G_\Phi \).

The following lemma provides a characterization of the conditions under which a Hermitian matrix can arise as a Gram matrix corresponding to a set of pure states.

	\begin{lemma}[see Ref.~\cite{chefles2004existence}]\label{lemma: Gram matrix realization}
		Let \( H \) be any candidate \( n \times n \) Hermitian matrix. Then, \( H \) is positive semidefinite with principal diagonal entries \( h_{ii} = 1 \) if and only if there exists some \( d \geq 2 \) (which depends on \( H \)) and some \( \Psi \in \mathbf{P}_d^n \) such that \( H = G_\Psi \).
			
	\end{lemma}
Let \( \mathbf{H}_n \) denote the set of all \( n \times n \) Hermitian matrices \( H \) that are positive semidefinite and have \( h_{ii} = 1 \) for all \( i = 1, \dots, n \). Based on this lemma, one can easily conclude that
\[
\mathcal{B}_n = \left\{ h_{12} h_{13} \cdots h_{n1} \mid H = (h_{kl}) \in \mathbf{H}_n \right\}.
\]
Throughout this paper, for \( k, l \in \{1, 2, \dots, n\} \), we will use the notation \( k \oplus l \) to denote the number
\[
k \oplus l = \begin{cases} 
	k + l & \text{if } k + l \leq n, \\
	k + l - n & \text{if } k + l > n.
\end{cases}
\]
Under this notation, the expression \( h_{12} h_{13} \cdots h_{n1} \) is equivalent to
\(
\prod_{j=1}^{n} h_{j (j \oplus 1)}.
\)
A \( n \times n \) matrix is circulant if it has the form
\begin{equation}\label{eq:circulant}
	G_{\mathbf{z}} =
	\begin{bmatrix}
		z_0 & z_1 & z_2 & \cdots & z_{n-2} & z_{n-1} \\
		z_{n-1} & z_0 & z_1 & \cdots & z_{n-3} & z_{n-2} \\
		z_{n-2} & z_{n-1} & z_0 & \cdots & z_{n-4} & z_{n-3} \\
		\vdots & \vdots & \vdots & \ddots & \vdots & \vdots \\
		z_2 & z_3 & z_4 & \cdots & z_0 & z_1 \\
		z_1 & z_2 & z_3 & \cdots & z_{n-1} & z_0
	\end{bmatrix},
\end{equation}
where \( \mathbf{z} = (z_0, z_1, \dots, z_{n-1}) \in \mathbb{C}^n \).
We denote the set of all \( n \times n \) circulant matrices as \( \mathbf{C}_n \). It follows that
\[
\mathcal{B}_{n \mid \text{circ}} = \left\{ \prod_{j=1}^{n} h_{j (j \oplus 1)} \mid H = (h_{kl}) \in \mathbf{H}_n \cap \mathbf{C}_n \right\}.
\]
Let \( C_n = G_{(0,1,0,\dots,0)} \), i.e., \( C_n \) is the special \( n \times n \) circulant matrix given by
\[
C_n =
\begin{bmatrix}
	0 & 1 & 0 & \cdots & 0 \\
	0 & 0 & 1 & \cdots & 0 \\
	\vdots & \vdots & \ddots & \ddots & \vdots \\
	0 & 0 & 0 & \ddots & 1 \\
	1 & 0 & 0 & \cdots & 0
\end{bmatrix}.
\]
Then, the circulant matrix \( G_{\mathbf{z}} \) defined in Eq. \eqref{eq:circulant} can be written as
\[
G_{\mathbf{z}} = \sum_{j=0}^{n-1} z_j C_n^j.
\]
Thus, the Bargmann invariant set is
\begin{equation}\label{eq:B_n_circ}
	\mathcal{B}_{n \mid \text{circ}} = \left\{ z_1^n \mid \mathbb{I} + \sum_{j=1}^{n-1} z_j C_n^j \in \mathbf{H}_n \right\}.
\end{equation}

\section{Properties of  Bargmann invariant sets}\label{sect:Main}

\begin{figure*}[t]
	\centering
	\includegraphics[width=1.9\columnwidth]{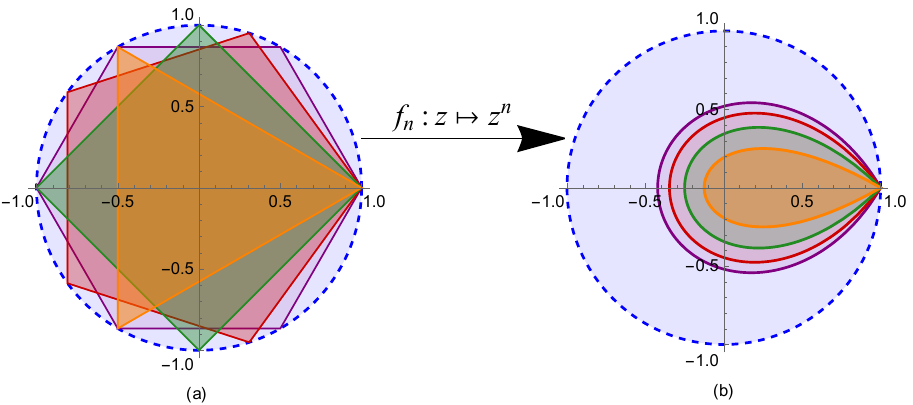}
	\caption{	The set $ \mathcal{P}_n \text{ maps onto the set } \mathcal{B}_{n \mid \text{circ}} \text{ via the } n\text{-th power map } f_n: z \mapsto z^n$.  {(a). The regions of $\mathcal{P}_n \text{ for } n = 3, 4, 5, 6.$ \quad  (b). Sets of quantum-realizable Bargmann invariants of   $\mathcal{B}_{n \mid \text{circ}} \text{ for } n = 3, 4, 5, 6. $		} 
	 }
	\label{fig:tri_square}
\end{figure*}

Note that \( \mathcal{B}_1 = \{1\} \), \( \mathcal{B}_2 = [0, 1] \), and \( \mathcal{B}_3 = \left\{ |\Delta| e^{\mathrm{i}\phi} \in \mathbb{C} \mid 1 - 3 |\Delta|^{\frac{2}{3}} + 2 |\Delta| \cos(\phi) \geq 0 \right\} \), as revealed in Ref. \cite{PhysRevLett.133.190201}. It was conjectured that the set \( \mathcal{B}_4 = \mathcal{B}_{4 \mid \text{circ}} \), where the boundary \( \Delta \) of \( \mathcal{B}_{4 \mid \text{circ}} \) is characterized by the following points:
\[
\Delta = \frac{e^{\mathrm{i}\phi}}{(\sin(\phi/4) + \cos(\phi/4))^4}.
\]

Although \( \mathcal{B}_3 \) has been known, it has not been established whether \( \mathcal{B}_3 = \mathcal{B}_{3 \mid \text{circ}} \). To address this, we first need to characterize the set \( \mathcal{B}_{3 \mid \text{circ}} \). For a circulant matrix \( \mathbb{I} + z_1 C_3 + z_2 C_3^2 \) to be positive semidefinite, it must take the following form (i.e., \( z_2 = \overline{z}_1 \)):
\[
G_{(1, z_1, \overline{z}_1)} = \begin{bmatrix}
	1 & z_1 & \overline{z}_1 \\
	\overline{z}_1 & 1 & z_1 \\
	z_1 & \overline{z}_1 & 1
\end{bmatrix}.
\]
The eigenvalues of the above matrix are:
\[
1 + z_1 + \overline{z}_1, \quad 1 + z_1 \omega + \overline{z}_1 \omega^2, \quad 1 + z_1 \omega^2 + \overline{z}_1 \omega,
\]
where \( \omega = e^{\frac{2\pi \mathrm{i}}{3}} \). Let \( z_1 = x +  \mathrm{i}y \). To ensure that these eigenvalues are nonnegative, the following conditions must hold:
\begin{equation}\label{eq:R_3eq}
	1 + 2x \geq 0, \quad 1 - x - \sqrt{3}y \geq 0, \quad 1 - x + \sqrt{3}y \geq 0.
\end{equation}

Denote the triangle defined by these inequalities as \( \mathcal{P}_3 \) (see Fig. \ref{fig:tri_square}). We then find that
\[
\mathcal{B}_{3 \mid \text{circ}} = \{ z_1^3 \mid z_1 \in \mathcal{P}_3 \}.
\]
It is easy to check that the set \( \{ z_1^3 \mid z_1 \in \mathcal{P}_3 \} \) is exactly \( \left\{ |\Delta| e^{i\phi} \in \mathbb{C} \mid 1 - 3 |\Delta|^{\frac{2}{3}} + 2 |\Delta| \cos(\phi) \geq 0 \right\} \). Since \( \mathcal{B}_3 \) is characterized by the latter set \cite{PhysRevLett.133.190201}, we conclude that \( \mathcal{B}_3 = \mathcal{B}_{3 \mid \text{circ}} \).

Next, we confirm this result using a different method. Since we always have \( \mathcal{B}_{3 \mid \text{circ}} \subseteq \mathcal{B}_3 \), we need only check the reverse inclusion \( \mathcal{B}_3 \subseteq \mathcal{B}_{3 \mid \text{circ}} \).

For each \( z \in \mathcal{B}_3 \), there exists a matrix \( H = (h_{kl}) \in \mathbf{H}_3 \) such that
\[
z = \prod_{k=1}^3 h_{k (k \oplus 1)}.
\]
Without loss of generality, we assume that
\[
H = \begin{bmatrix}
	1 & r_{12} & r_{13} e^{-\mathrm{i}\phi} \\
	r_{12} & 1 & r_{23} \\
	r_{13} e^{\mathrm{i}\phi} & r_{23} & 1
\end{bmatrix}.
\]
Thus, \( z = r_{12} r_{23} r_{31} e^{\mathrm{i} \phi} \). Since \( H \in \mathbf{H}_3 \), we have \( r_{12}^2, r_{13}^2, r_{23}^2 \leq 1 \) and
\begin{equation}\label{eq:DetH3}
	\mathrm{Det}[H] = 1 + 2 r_{12} r_{13} r_{23} \cos \phi - r_{12}^2 - r_{13}^2 - r_{23}^2 \geq 0.
\end{equation}
Let
\[
r = (r_{12} r_{23} r_{31})^{\frac{1}{3}} \quad \text{and} \quad C_H := \begin{bmatrix}
	1 & r e^{ \mathrm{i} \frac{\phi}{3}} & r e^{- \mathrm{i} \frac{\phi}{3}} \\
	r e^{- \mathrm{i} \frac{\phi}{3}} & 1 & r e^{\textbf{i} \frac{\phi}{3}} \\
	r e^{ \mathrm{i} \frac{\phi}{3}} & r e^{- \mathrm{i} \frac{\phi}{3}} & 1
\end{bmatrix}.
\]
The matrix \( C_H \) belongs to \( \mathbf{H}_3 \) if and only if \( r \leq 1 \) and
\begin{equation}\label{eq:DetCH3}
	\mathrm{Det}[C_H] = 1 + 2 r^3 \cos \phi - 3 r^2 \geq 0.
\end{equation}
Clearly, the first condition \( r = (r_{12} r_{23} r_{31})^{\frac{1}{3}} \leq 1 \) holds. 
For the second condition, we consider
\[
\mathrm{Det}[C_H] - \mathrm{Det}[H] = r_{12}^2 + r_{13}^2 + r_{23}^2 - 3 (r_{12} r_{23} r_{31})^{\frac{2}{3}} \geq 0.
\]
The last inequality holds by applying the well-known inequality 
$
\frac{a + b + c}{3} \geq (abc)^{\frac{1}{3}} \quad \text{for} \quad a, b, c \geq 0.
$
Thus, we conclude that
\[
\mathrm{Det}[C_H] \geq \mathrm{Det}[H] \geq 0.
\]
Hence, \( C_H \in \mathbf{H}_3 \). Since \( C_H = \mathbb{I} + z_1 C_3 + z_2 C_3^2 \) where \( z_1 = r e^{\mathrm{i} \frac{\phi}{3}} = \overline{z}_2 \), we have
\[
z = \left( r e^{\mathrm{i} \frac{\phi}{3}} \right)^3 = z_1^3 \in \mathcal{B}_{3 \mid \text{circ}}.
\]
Therefore, \( \mathcal{B}_3 \subseteq \mathcal{B}_{3 \mid \text{circ}} \). To conclude, we have \( \mathcal{B}_3 = \mathcal{B}_{3 \mid \text{circ}} \).

Thus, the Bargmann invariant set \( \mathcal{B}_3 \) is exactly characterized by
\[
\mathcal{B}_3 = f_3(\mathcal{P}_3) = \left\{ |\Delta| e^{\mathrm{i}\phi} \in \mathbb{C} \mid 1 - 3 |\Delta|^{\frac{2}{3}} + 2 |\Delta| \cos(\phi) \geq 0 \right\}.
\]
We have thus reproduced Theorem 1 of Ref. \cite{PhysRevLett.133.190201} by an alternative method.

Now, we seek to characterize the set \(\mathcal{B}_{n \mid \text{circ}}\) for \(n = 4\). Consider a positive semidefinite \(4 \times 4\) circulant matrix, which must take the following form:

\[
G_{(1, z_1, z_2, \overline{z}_1)} = \begin{bmatrix}
	1 & z_1 & z_2 & \overline{z}_1 \\
	\overline{z}_1 & 1 & z_1 & z_2 \\
	z_2 & \overline{z}_1 & 1 & z_1 \\
	z_1 & z_2 & \overline{z}_1 & 1
\end{bmatrix},
\]
where \(z_1 = x +  \mathrm{i} y \in \mathbb{C}\) and \(z_2 \in \mathbb{R}\). The eigenvalues of \(G_{(1, z_1, z_2, \overline{z}_1)}\) are:
$$
\begin{array}{rcl}
	\lambda_1 &=& 1 + z_1 + z_2 + \overline{z}_1 = 1 + 2x + z_2 \geq 0, \\[2mm]
	\lambda_2 &=& 1 + \mathrm{i} z_1 - z_2 - \mathrm{i} \overline{z}_1 = 1 - 2y - z_2 \geq 0, \\[2mm]
	\lambda_3 &=& 1 - z_1 + z_2 - \overline{z}_1 = 1 - 2x + z_2 \geq 0, \\[2mm]
	\lambda_4 &=& 1 - \mathrm{i} z_1 - z_2 + \mathrm{i} \overline{z}_1 = 1 + 2y - z_2 \geq 0.
\end{array}
$$
From these conditions, we deduce that:
$$
|x| + |y| \leq 1,
$$
which describes a square \(\mathcal{P}_4\) (see Fig. \ref{fig:tri_square}) centered at \(0\) with \(1\) as one of its vertices. Therefore, we find that:
\[
\mathcal{B}_{4 \mid \text{circ}} = \{z_1^4 \mid z_1 \in \mathcal{P}_4\}.
\]
This result regarding the characterization of \(\mathcal{B}_{3 \mid \text{circ}}\) and \(\mathcal{B}_{4 \mid \text{circ}}\) can be generalized to \(\mathcal{B}_{n \mid \text{circ}}\) for higher \(n\).

Let \(\mathcal{P}_n\) denote the region enclosed by a regular \(n\)-sided polygon centered at the origin with one vertex at \(1\) in the complex plane. Define the map \(f_n: \mathbb{C} \to \mathbb{C}\) by
\[
f_n(z) = z^n.
\]
Next, let
\[
\mathcal{Z}_n = \left\{ z_1 \mid \mathbb{I} + \sum_{j=1}^{n-1} z_j C_n^j \in \mathbf{H}_n \right\}.
\]
By Eq. \eqref{eq:B_n_circ}, we have \(\mathcal{B}_{n \mid \text{circ}} = \{ z_1^n \mid z_1 \in \mathcal{Z}_n \}\). Thus, we characterize \(\mathcal{B}_{n \mid \text{circ}}\) as the image of the regular polygon \(\mathcal{P}_n\) under \(f_n\), as follows.

\begin{theorem} \label{theorem: polygon}
	For each \(n \geq 3\), the Bargmann invariant set \(\mathcal{B}_{n \mid \emph{circ}}\) is exactly the image of the set \(\mathcal{P}_n\) under the map \(f_n\), i.e.,	
	\[
	\mathcal{B}_{n \mid \emph{circ}} = f_n(\mathcal{P}_n) := \{ f_n(z) \mid z \in \mathcal{P}_n \} = \{ z^n \mid z \in \mathcal{P}_n \}.
	\]
\end{theorem}

\noindent \emph{Sketch of the proof:} To prove this, we need to show that \(\mathcal{Z}_n = \mathcal{P}_n\). We prove this equality in four steps, which are detailed in Appendix \ref{ap:Proof_Thm1}.

 \textbf{Step 1:} Show that \(1 \in \mathcal{Z}_n\).
 
 \textbf{Step 2:} If \(z_1 \in \mathcal{Z}_n\), then \(\xi z_1 \in \mathcal{Z}_n\) (where \(\xi = e^{\frac{2\pi \mathrm{i}}{n}}\)).
 
 \textbf{Step 3:} Prove that \(\mathcal{Z}_n\) is a convex set.
 
  \textbf{Step 4:} Show that for each \(z_1 = x + \mathrm{i}y \in \mathcal{Z}_n\),  the point $(x,y)$ must below the line $\ell_n:$
	\[
	x \cos \frac{\pi}{n} + y \sin \frac{\pi}{n} - \cos \frac{\pi}{n} =0,
	\]
	which corresponds to a line passing through the points \(1\) and \(\xi\).

Steps 1 and 2 establish that \(\mathcal{P}_n \subseteq \mathcal{Z}_n\), while Steps 2 and 4 imply that \(\mathcal{Z}_n \subseteq \mathcal{P}_n\). Therefore, we conclude that \(\mathcal{Z}_n = \mathcal{P}_n\). \qed

\vskip 5pt
In fact, the image \( f_n(\mathcal{P}_n) \) is an \( n \)-fold cover of \( \mathcal{B}_{n \mid \text{circ}} \). We can decompose the polygon \( \mathcal{P}_n \) into \( n \) triangles, \( \mathcal{T}_1, \cdots, \mathcal{T}_n \), such that \( \mathcal{B}_{n \mid \text{circ}} = f_n(\mathcal{T}_j) \) for each \( j \). For instance, in Fig. \ref{fig:Three_points}, the triangle \( \mathcal{T}_1 \) is plotted in green, and the image of the green edge under the map \( f_n \) is exactly the green curve \( \partial \mathcal{B}_{n \mid \text{circ}} \), which represents the boundary of \( \mathcal{B}_{n \mid \text{circ}} \).

\begin{figure}[h]
	\centering
	\includegraphics[width=1\columnwidth]{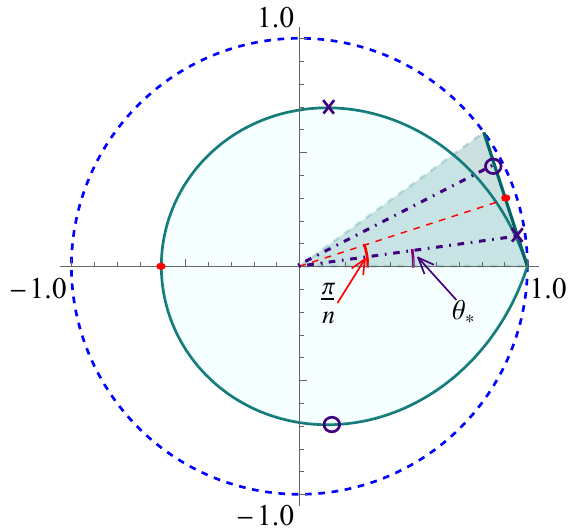}
	\caption{Three special points on the boundary \( \partial \mathcal{B}_{n \mid \text{circ}} \) of the Bargmann invariant set \( \mathcal{B}_{n \mid \text{circ}} \) are considered: the point labeled by ‘\( \times \)’ or ‘\( \circ \)’ corresponds to the maximal imaginarity \( I_n \), while the red bullet {‘\color{red}{$\bullet$}}’	 corresponds to the minimal real number in \( \mathcal{B}_{n \mid \text{circ}} \). Additionally, we indicate how these three points arise from the map \( f_n: z \mapsto z^n \), which is applied to the edge \( \ell_n \) (the line connecting \( 1 \) and \( \xi \)) of \( \mathcal{P}_n \).
	}
	\label{fig:Three_points}
\end{figure}

\vskip 5pt

We can also express the boundary curve of \( \mathcal{B}_{n \mid \text{circ}} \) in polar coordinates. Let \( \xi = e^{\frac{2\pi \mathrm{ i}}{n}} \) and \( c = r e^{\mathrm{i}\phi} \in \partial \mathcal{B}_{n \mid \text{circ}} \). Then there exists \( t \in [0,1] \) such that
\[
c = \left( t \times 1 + (1-t) \xi \right)^n = \left( \sin^2 \varphi + \xi \cos^2 \varphi \right)^n,
\]
where \( t = \sin^2 \varphi \).
The radius \( r = |c| \) is given by
 \begin{equation}
	\begin{array}{rcl}
		r^2&=& \displaystyle \left(\sin^{4}\varphi+\cos^{4} \varphi +2\cos \frac{2\pi}{n}  \sin^2 \varphi \cos^2\varphi\right)^{n}\\[3mm]
		&=& \displaystyle \left( 1 +2(\cos \frac{2\pi}{n} -1) \sin^2 \varphi \cos^2\varphi \right)^{n}\\[3mm]
		&=& \displaystyle \left( 1 -4 \sin ^2 \frac{ \pi}{n}   \sin^2 \varphi \cos^2\varphi \right)^{n}. 
	\end{array}
\end{equation}
The phase \( \phi \) satisfies
 \begin{equation}\label{eq:r_phi}
\tan \frac{\phi}{n} = \frac{\sin  \frac{2\pi}{n}  \cos^2 \varphi}{\sin^2 \varphi + \cos  \frac{2\pi}{n}  \cos^2 \varphi}
= \frac{\sin  \frac{2\pi}{n}   \cos^2 \varphi}{1 - (1 - \cos  \frac{2\pi}{n}  ) \cos^2 \varphi}.
\end{equation}

From this equation, we deduce that
\[
\cos^2 \varphi = \frac{\tan \frac{\phi}{n}}{\tan \frac{\phi}{n} \left( 1 - \cos  \frac{2\pi}{n}   \right) + \sin  \frac{2\pi}{n}  }.
\]
Substituting this into Eq. \eqref{eq:r_phi}, we obtain the relationship between \( r \) and \( \phi \).

For a complex number \( z = x + \mathrm{i} y \in \mathbb{C} \), define
\[
\mathbf{I}(z) = |\mathrm{Im}(z)| = |y|,
\]
which measures the imaginary part of \( z \). It is interesting to find the maximal imaginarity contained in \( \mathcal{B}_{n \mid \text{circ}} \). We observe that
\[
I_n := \max_{z \in \mathcal{B}_{n \mid \text{circ}}} \mathbf{I}(z) = \frac{\cos^n \left( \frac{\pi}{n} \right)}{\cos^n \left( \frac{\pi}{n} - \theta_* \right)} \sin n \theta_*
\]
where \( \theta_* = \frac{\frac{\pi}{2} - \frac{\pi}{n}}{n-1} = \frac{(n-2)\pi}{2n(n-1)} \) (see Fig. \ref{fig:Three_points} for an intuitive representation of the value of \( I_n \)).
In fact, we only need to consider \( \mathbf{I}(z) \) for \( z \in \partial \mathcal{B}_{n \mid \text{circ}} \). Note that
\[
\partial \mathcal{B}_{n \mid \text{circ}} = \left\{ z_1^n \mid z_1 = \sin^2 \varphi + \xi \cos^2 \varphi, \, \varphi \in \left[ 0, \frac{\pi}{2} \right] \right\}.
\]

Writing \( z_1 \) in the form \( r e^{\mathrm{i} \theta} \), and referring to the geometry shown in Fig. \ref{fig:Three_points}, we have the relation
\(
r \cos \left( \frac{\pi}{n} - \theta \right) = \cos \frac{\pi}{n}.
\)
Therefore, 
\[
\mathbf{I}(z) = \mathbf{I}(z_1^n) = r^n \sin n \theta = \frac{\cos^n \frac{\pi}{n}}{\cos^n \left( \frac{\pi}{n} - \theta \right)} \sin n \theta.
\]
It is natural to define the function
\[
f(\theta) = \frac{\sin n \theta}{\cos^n \left( \frac{\pi}{n} - \theta \right)}.
\]
We then find that
\[
f'(\theta) = \frac{n \cos^{n-1} \left( \frac{\pi}{n} - \theta \right)}{\cos^{2n} \left( \frac{\pi}{n} - \theta \right)} \cos \left[ (n-1) \theta + \frac{\pi}{n} \right].
\]

Setting \( f'(\theta_*) = 0 \), we obtain \( (n-1)\theta_* +\frac{\pi}{n}  = \frac{\pi}{2} \), which corresponds to an optimal value \( \mathbf{I}(z_*) \) (see the points label with $\times$ in Fig. \ref{fig:Three_points}). Moreover, the relationship between \( \varphi \) and \( \theta \) is given by
\[
\cos^2 \varphi = \frac{\cos \frac{\pi}{n} \sin \theta}{\sin \frac{2\pi}{n} \cos \left( \frac{\pi}{n} - \theta \right)} = \frac{\sin \theta}{2 \sin \frac{\pi}{n} \cos \left( \frac{\pi}{n} - \theta \right)}.
\]
One can also determine the optimal parameter \( \varphi_* \), which may be useful for experiments.

\vskip 5pt

We are primarily interested in understanding the nature of the set \(\mathcal{B}_n\). First, by definition, we have the following inclusion:
\(
\mathcal{B}_{n \mid \text{circ}} \subseteq \mathcal{B}_n.
\)
Naturally, one might ask whether the reverse inclusion holds, i.e., whether we have
\[
\mathcal{B}_n \subseteq \mathcal{B}_{n \mid \text{circ}}.
\]

For each \(H = (h_{kl}) \in \mathrm{Mat}_{n \times n}(\mathbb{C})\), we can associate it with a circulant matrix \(C_H = \sum_{k=0}^{n-1} z_k C_n^k\), where the coefficients are given by
\[
z_k = \left( \prod_{l=1}^n r_{l(l \oplus k)} \right)^{\frac{1}{n}} e^{\mathrm{i} \theta_k},
\]
with \(\theta_k := \frac{\sum_{l=1}^n \theta_{l(l \oplus k)}}{n}\) and \(h_{kl} = r_{kl} e^{\mathrm{i} \theta_{kl}}\), where \(r_{kl} \geq 0\) and \(\theta_{kl} \in (-\pi, \pi]\). It is important to note that the Bargmann invariants corresponding to \(H\) and \(C_H\) are equal, i.e.,
\[
\prod_{j=1}^n h_{j(j \oplus 1)} = z_1^n.
\]
For the case when \(n = 3\), we deduced that the positive semidefiniteness of \(C_H\) follows from the positive semidefiniteness of \(H\). This leads us to the conjecture that the positive semidefiniteness of \(H\) implies the positive semidefiniteness of \(C_H\) (which may depend on \(n\)) for general \(n\). If this conjecture holds, we would then have \(\mathcal{B}_n = \mathcal{B}_{n \mid \text{circ}}\).

For the time being, we have not been able to prove this strong result. Instead, we present some properties of the set \(\mathcal{B}_n\) in the following propositions (see Proposition \ref{prop:Hn}, \ref{prop:Bn_Surperset}, and \ref{prop:Bnd_Starset}).

Let us first recall the definition of the Hadamard product \(\circ\) of two matrices. Given two matrices \(G = (g_{kl})\) and \(H = (h_{kl})\) in \(\mathrm{Mat}_{n \times n}(\mathbb{C})\), the Hadamard product \(G \circ H\) is also an \(n \times n\) matrix whose \(kl\)-th entry is \(g_{kl} h_{kl}\).

\begin{proposition}\label{prop:Hn}
	If \(H_1, H_2 \in \mathbf{H}_n\), then their Hadamard product \(H_1 \circ H_2 \in \mathbf{H}_n\). As a consequence, \(\mathcal{B}_n\) is closed under multiplication. Specifically, if \(c_1, c_2 \in \mathcal{B}_n\), then \(c_1 c_2 \in \mathcal{B}_n\).
\end{proposition}
\begin{proof}
	A matrix \(H \in \mathbf{H}_n\) if and only if there exists a set of quantum states \(\Psi = (|\psi_1\rangle, \cdots, |\psi_n\rangle)\) such that \(H = G_\Psi\). Suppose that \(H_1 = G_\Psi\) and \(H_2 = G_\Phi\), where \(\Psi = (|\psi_1\rangle, \cdots, |\psi_n\rangle)\) and \(\Phi = (|\phi_1\rangle, \cdots, |\phi_n\rangle)\). Then we have the tensor product of two sets of quantum states:
	\[
	\Psi \otimes \Phi := \{|\psi_1\rangle \otimes |\phi_1\rangle, \cdots, |\psi_n\rangle \otimes |\phi_n\rangle\}.
	\]
	It can be checked that
	\[
	H_1 \circ H_2 = G_\Psi \circ G_\Phi = G_{\Psi \otimes \Phi} \in \mathbf{H}_n.
	\]
	For the latter statement, if \(c_1, c_2 \in \mathcal{B}_n\), then there exist \(G = (g_{kl})\) and \(H = (h_{kl}) \in \mathbf{H}_n\) such that
	\[
	c_1 = \prod_{j=1}^n g_{j(j \oplus 1)}, \quad c_2 = \prod_{j=1}^n h_{j(j \oplus 1)}.
	\]
	By the previous argument, the Hadamard product \(G \circ H \in \mathbf{H}_n\), and thus the product
	\[
	c_1 c_2 = \prod_{j=1}^n (g_{j(j \oplus 1)} h_{j(j \oplus 1)}) \in \mathcal{B}_n.
	\]
	This completes the proof.
\end{proof}

Although we cannot prove the inclusion \(\mathcal{B}_n \subseteq \mathcal{B}_{n \mid \text{circ}}\) at this time, we can bound \(\mathcal{B}_n\) by a superset.

\begin{proposition}\label{prop:Bn_Surperset}
	Let \(n \geq 3\) be an integer. Then we have
	\[
	\mathcal{B}_n \subseteq \mathcal{P}_n.
	\]
\end{proposition}
\begin{proof}
	For any \(c \in \mathcal{B}_n\), there exists a Gram matrix \(H = (h_{kl})\) such that \(c = \prod_{j=1}^n h_{j(j \oplus 1)} = \prod_{j=1}^n \langle \psi_j | \psi_{j \oplus 1} \rangle\). For each \(k \in \{1, 2, \dots, n\}\), we define
	\[
	\Psi_k := (|\psi_k\rangle, |\psi_{k \oplus 1}\rangle, |\psi_{k \oplus 2}\rangle, \cdots, |\psi_{k \oplus (n-1)}\rangle).
	\]
	It can be verified that the product
	\[
	G_{\Psi_1} \circ G_{\Psi_2} \circ \cdots \circ G_{\Psi_n}
	\]
	is not only in \(\mathbf{H}_n\), but is also a circulant matrix. Specifically, if we define
	\(
	\mathbf{z} =  ( 1, z_1,   z_2 , \cdots, z_{n-1} )^T
	\)
	where \(z_k := \prod_{j=1}^n \langle \psi_j | \psi_{j \oplus k} \rangle\), for \(k \in \{1, 2, \dots, n-1\}\), then
	\[
	G_{\Psi_1} \circ G_{\Psi_2} \circ \cdots \circ G_{\Psi_n} = G_{\mathbf{z}}.
	\]
	By Proposition \ref{prop:Hn}, \(G_{\mathbf{z}} \in \mathbf{H}_n\), and therefore \(z_1 \in \mathcal{P}_n\). Since \(c = z_1\), it follows that \(c \in \mathcal{P}_n\). Thus, we conclude that
	\[
	\mathcal{B}_n \subseteq \mathcal{P}_n.
	\]
\end{proof}

 \begin{figure}[h]
 	\centering
 	\includegraphics[width=0.6\columnwidth]{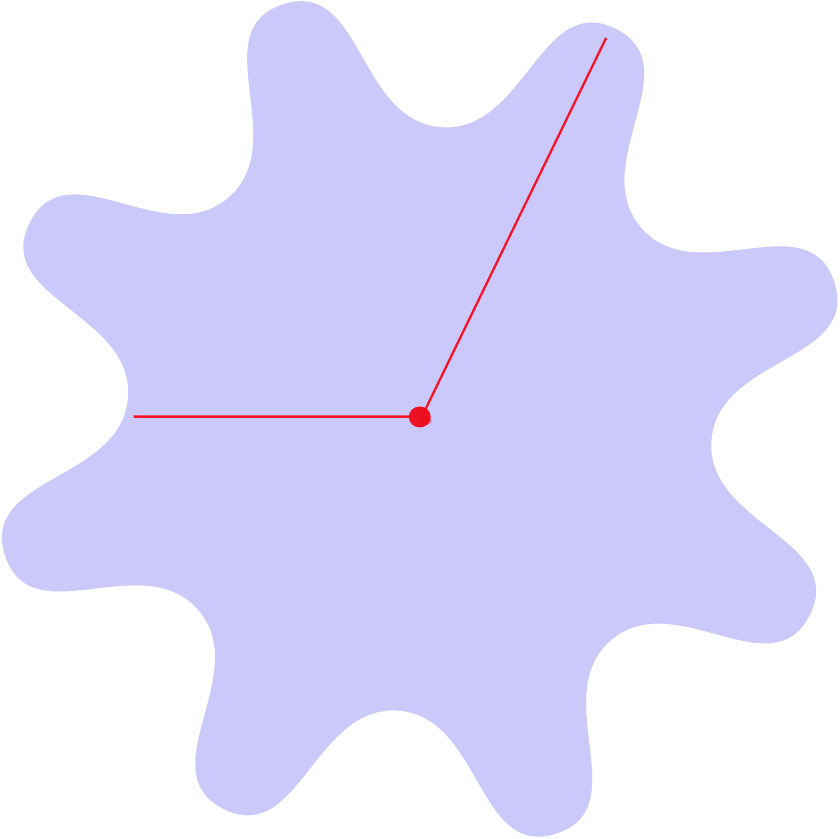}
 	\caption{The figure is a diagram of a set with star-shaped. Note that a set with star-shaped  is not necessarily a convex set.  }
 	\label{fig:star}
 \end{figure}
 
\begin{proposition}\label{prop:Bnd_Starset} 
	Let \(n \geq 2\) be an integer. Then the set \(\mathcal{B}_{n,d}\) is star-shaped with the center at \(0\) (see Fig. \ref{fig:star} for an intuitive representation of a star-shaped set). That is, if \(c \in \mathcal{B}_{n,d}\), then all points lying between \(0\) and \(c\) are also in \(\mathcal{B}_{n,d}\).
\end{proposition}

\begin{proof}
	For \(n = 2\), we know that \(\mathcal{B}_{2,d} = [0, 1]\), which is clearly a star-shaped set centered at \(0\). Thus, we only need to consider the case \(n \geq 3\). Clearly, \(0 \in \mathcal{B}_{n,d}\).
	
	For each \(c \in \mathcal{B}_{n,d} \setminus \{0\}\), there exists a set of quantum states \(\Psi = (|\psi_1\rangle, |\psi_2\rangle, \cdots, |\psi_n\rangle) \in \mathbf{P}_d^n\) such that
	\[
	c = \prod_{j=1}^n \langle \psi_j | \psi_{j \oplus 1} \rangle.
	\]
	We need to show that for each \(t \in [0, 1]\), the point \(tc = t c + (1 - t) 0 \in \mathcal{B}_{n,d}\).
	
	Since \(c \neq 0\), we have \(\langle \psi_j | \psi_{j \oplus 1} \rangle \neq 0\) for all \(j\). Let \(\langle \psi_1 | \psi_2 \rangle \langle \psi_2 | \psi_3 \rangle = r e^{\mathrm{i} \theta}\), where \(r > 0\) and \(\theta \in [0, 2\pi]\). If \(\theta \neq 0\), we replace \(|\psi_1\rangle\) by \(e^{\mathrm{i} \theta} |\psi_1\rangle\). Therefore, we can always assume that \(\langle \psi_1 | \psi_2 \rangle \langle \psi_2 | \psi_3 \rangle = r\), where \(r\) is a positive real number.
	
	Since \(d \geq 2\), there always exists a state orthogonal to \(|\psi_1\rangle\). Let this state be denoted as \(|\psi_1^\perp\rangle\). By choosing an appropriate global phase for \(|\psi_1^\perp\rangle\), we can also assume that \(\langle \psi_1^\perp | \psi_3 \rangle \langle \psi_1 | \psi_2 \rangle\) is a real number, which we denote by \(r^\perp\). Replacing the second state of \(\Psi\) by \(|\psi_1^\perp\rangle\), we obtain a new set of states \(\Psi^\perp = (|\psi_1\rangle, |\psi_1^\perp\rangle, |\psi_3\rangle, \cdots, |\psi_n\rangle)\).
	
	The states \(|\psi_1^\perp\rangle\) and \(|\psi_2\rangle\) are linearly independent. Otherwise, if \(|\psi_2\rangle \propto |\psi_1^\perp\rangle\), we would have \(\langle \psi_1 | \psi_2 \rangle = 0\). For each \(p \in [0, 1]\), we define a pure state
	\[
	|\psi(p)\rangle = \frac{p |\psi_1^\perp\rangle + (1 - p) |\psi_2\rangle}{N_p},
	\]
	where \(N_p\) is the norm of the vector \(p |\psi_1^\perp\rangle + (1 - p) |\psi_2\rangle\), which is nonzero due to the linear independence of \(|\psi_1^\perp\rangle\) and \(|\psi_2\rangle\). Note that \(\{|\psi(p)\rangle\}_p\) is a family of pure states that connect \(|\psi_1^\perp\rangle\) and \(|\psi_2\rangle\).
	
	Now, consider the set of states \[\Psi(p) = (|\psi_1\rangle, |\psi(p)\rangle, |\psi_3\rangle, \cdots, |\psi_n\rangle).\]The corresponding Bargmann invariant \(c(p)\) arising from \(\Psi(p)\) is
	\[
	c(p) = \langle \psi_1 | \psi(p) \rangle \langle \psi(p) | \psi_3 \rangle \prod_{j=3}^n \langle \psi_j | \psi_{j \oplus 1} \rangle.
	\]
Set
	\(
	f(p) := \langle \psi_1 | \psi(p) \rangle \langle \psi(p) | \psi_3 \rangle.
	\)
	Since \(\prod_{j=3}^n \langle \psi_j | \psi_{j \oplus 1} \rangle = \frac{c}{r}\), we have
	\[
	c(p) = \frac{c f(p)}{r}.
	\]
	Substituting the expression for \(\psi(p)\) into the definition of \(f(p)\), we get
	\[
	f(p) = \frac{1}{N_p^2} \left[ (1 - p) \langle \psi_1 | \psi_2 \rangle \right] \left[ p \langle \psi_1^\perp | \psi_3 \rangle + (1 - p) \langle \psi_2 | \psi_3 \rangle \right].
	\]
	This simplifies to
	\[
	f(p) = \frac{1 - p}{N_p^2} \left[ p r^\perp + (1 - p) r \right] \in \mathbb{R}.
	\]
	Therefore, \(f(p)\) is a continuous real-valued function on \([0, 1]\) with \(f(0) = r\) and \(f(1) = 0\). Since \(f(1) = 0 \leq r t \leq r = f(0)\) for all \(t \in [0, 1]\), by the mean value theorem for continuous functions, there exists some \(p_t \in [0, 1]\) such that
	\[
	f(p_t) = r t.
	\]
	For such \(p_t\), we have
	\[
	c(p_t) = \frac{c f(p_t)}{r} = \frac{c r t}{r} = t c.
	\]
	This shows that \(t c \in \mathcal{B}_{n,d}\) for all \(t \in [0, 1]\), completing the proof.
\end{proof}

	\section{Quantum realization of Bargmann invariants in   qubit system}\label{sect:quantum realization}
	
 In the remark for the proof of Theorem \ref{theorem: polygon}, we pointed out that each point on the boundary \( \partial \mathcal{B}_{n \mid \text{circ}} \) of \( \mathcal{B}_{n \mid \text{circ}} \) has a quantum realization in a qubit system, i.e., \( \partial \mathcal{B}_{n \mid \text{circ}} \subseteq \mathcal{B}_{n,2} \). Here, we provide a more constructive proof of this statement and a stronger result.

	\begin{figure}[h]
		\centering
		\includegraphics[width=\columnwidth]{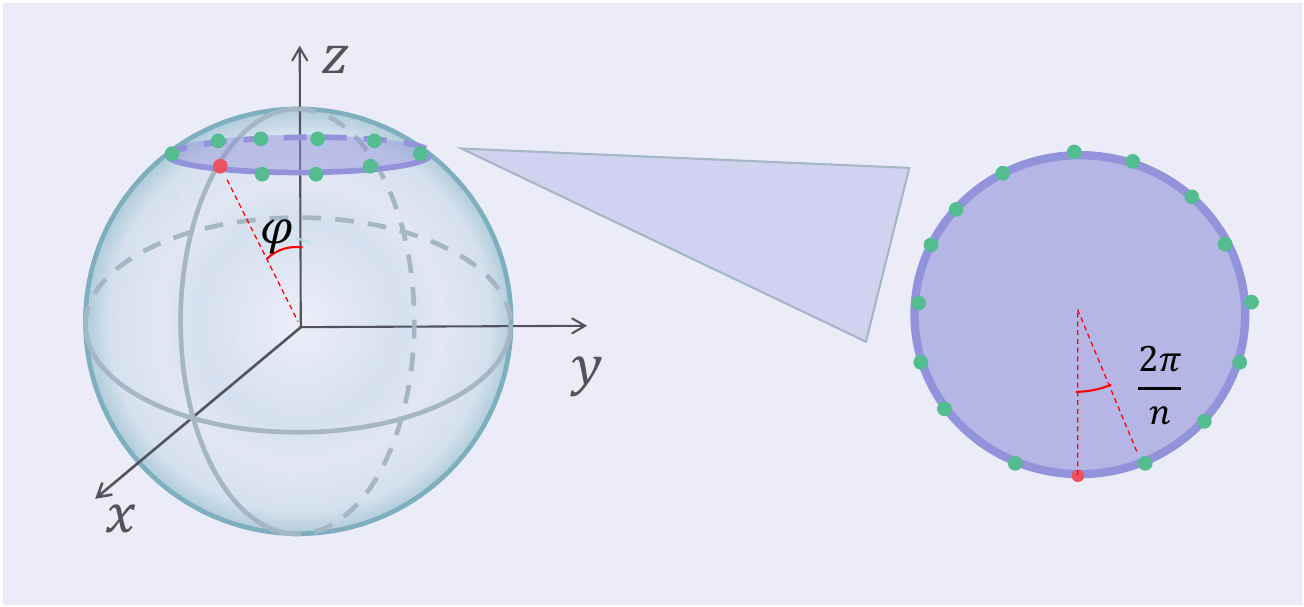}
		\caption{ A visualization of the \(n\)-tuple states \(\Psi(\varphi) = (|\psi_0 \rangle, |\psi_1 \rangle, \cdots, |\psi_{n-1} \rangle)\), as given by Eq. \eqref{eq:qubit_state}, in the context of the Bloch sphere. 
		}\label{fig:states}
	\end{figure}

	 \begin{theorem}\label{theorem: Realization}
	 	For each \( n \geq 3 \), any element \( z \) in the Bargmann invariant set \( \mathcal{B}_{n \mid \emph{circ}} \) can be realized in a qubit system. That is, 
	 	\[
	 	\mathcal{B}_{n \mid \emph{circ}} \subseteq \mathcal{B}_{n,2}.
	 	\]
	 \end{theorem}
	 
	 \begin{proof}
	 	First, observe that \( \mathcal{B}_{n \mid \text{circ}} \) represents the region whose boundary is characterized by the curve:
	 	\[
	 	\partial \mathcal{B}_{n \mid \text{circ}} = \{ z_1^n \mid z_1 = t \times 1 + (1-t) \xi, t \in [0,1] \},
	 	\]
	 	where \( \xi = e^{\frac{2\pi \mathrm{i}}{n}} \). We claim that 
	 	\( \partial \mathcal{B}_{n \mid \text{circ}} \subseteq \mathcal{B}_{n,2} \).	 	
	 	  \( \forall \varphi \in [0, 2\pi] \), set \( \Psi(\varphi) = (|\psi_0 \rangle, |\psi_1 \rangle, \cdots, |\psi_{n-1} \rangle) \) as an \( n \)-tuple of single-qubit states, where (see Fig. \ref{fig:states})
	 	\begin{equation}\label{eq:qubit_state}
	 		\left|\psi_k\right\rangle = \sin\varphi |0\rangle + \xi^k \cos\varphi |1\rangle, 
	 	\end{equation}
	 	for \( k = 0, 1, 2, \dots, n-1 \). The \( n \)-order Bargmann invariant of \( \Psi(\varphi) \) is given by
	 	\[
	 	c = \prod_{j=0}^{n-1} \langle \psi_j | \psi_{j \oplus 1} \rangle = \left( \sin^2 \varphi + \xi \cos^2 \varphi \right)^n \in \mathcal{B}_{n,2}.
	 	\]
	 	Let \( \sin^2 \varphi = t \), so \( \cos^2 \varphi = 1 - t \). Therefore, for each \( t \in [0,1] \),
	 	\[
	 	\left( t \times 1 + (1-t) \xi \right)^n \in \mathcal{B}_{n,2}.
	 	\]
	 	Thus, we have \( \partial \mathcal{B}_{n \mid \text{circ}} \subseteq \mathcal{B}_{n,2} \).
	 	
	 	By Proposition \ref{prop:Bnd_Starset}, \( \mathcal{B}_{n,2} \) contains the set of all points connecting 0 to some point on \( \partial \mathcal{B}_{n \mid \text{circ}} \). Therefore, we conclude that
	 $$
	 	\mathcal{B}_{n \mid \text{circ}} \subseteq \mathcal{B}_{n,2}.	 $$
	 \end{proof}
	 
	 \vskip 5pt
	 
	 If the value \( \mathrm{Tr}[ \psi_1 \psi_2 \cdots \psi_n] \) contains a nontrivial imaginary part, one can conclude the imaginarity of the \( n \)-tuple set \( \Psi = (|\psi_1 \rangle, |\psi_2 \rangle, \cdots, |\psi_{n} \rangle) \). It is of interest to consider whether there exist any real numbers of the form \( \mathrm{Tr}[ \psi_1 \psi_2 \cdots \psi_n] \) that could only arise from some quantum state set with imaginarity. This motivates us to define
	 \[
	 \mathcal{B}_{n,d}^{\mathbb{R}} := \left\{ \mathrm{Tr}[ \psi_1 \psi_2 \cdots \psi_n] \mid \psi_i \in \mathbf{P}_d \cap \mathrm{Mat}_{d \times d}(\mathbb{R}), \forall i \right\}.
	 \]
	 
	 \begin{theorem}\label{theorem: Real}
	 	For each \( n \geq 3 \), we have the following relation:
	 	\[
	 	\mathcal{B}_{n,2}^{\mathbb{R}} = \mathcal{B}_{n|\emph{circ}} \cap \mathbb{R} = [-\cos^n \frac{\pi}{n}, 1].
	 	\]
	 	As a consequence, for each \( n \)-order real Bargmann invariant   \( c \in \mathcal{B}_{n \mid \emph{circ}} \), there exists a realization of \( c \) using real states alone.
	 \end{theorem}
	 
	 \begin{proof}
	 	First, we show \( \mathcal{B}_{n,2}^{\mathbb{R}} \subseteq [-\cos^n \frac{\pi}{n}, 1] \). For any \( c \in \mathcal{B}_{n,2}^{\mathbb{R}} \), there exists \( \Psi = (|\psi_1 \rangle, |\psi_2 \rangle, \dots, |\psi_n \rangle) \in \mathbf{P}_2^n \cap (\mathbb{R}^2)^n \) such that
	 	\[
	 	c = \prod_{j=1}^n \langle \psi_j | \psi_{j \oplus 1} \rangle.
	 	\]
	 	Since each \( |\psi_j \rangle \in \mathbf{P}_2 \cap \mathbb{R}^2 \), it can be written as
	 	\[
	 	|\psi_j \rangle = \cos \varphi_j |0\rangle + \sin \varphi_j |1\rangle.
	 	\]
	 	Thus,
	 $
	 	c = \prod_{j=1}^n (\cos \varphi_j \cos \varphi_{j \oplus 1} + \sin \varphi_j \sin \varphi_{j \oplus 1}).
	 $
	 	This simplifies to
	 	\(
	 	c = \prod_{j=1}^n \cos (\varphi_j - \varphi_{j \oplus 1}).
	 	\)
	 	To find the maximal and minimal values of the function
	 	\[
	 	f(\varphi_1, \dots, \varphi_n) = \prod_{j=1}^n \cos (\varphi_j - \varphi_{j \oplus 1}),
	 	\]
	 	we compute the partial derivatives, leading to the equation \( \tan(\varphi_{j \ominus 1} - \varphi_j) = \tan(\varphi_j - \varphi_{j \oplus 1}) \).  Therefore, $\varphi_{j\ominus 1} -  \varphi_{j } +k_j \pi= \varphi_{j } -  \varphi_{j\oplus 1}  $.  That is, $\varphi_{j\ominus 1} -  \varphi_{j } \equiv  \varphi_{j } -  \varphi_{j\oplus 1}  \text { mod } \pi$ (here, $a\equiv b  \text { mod } \pi $ if and only if $a-b \in \mathbb{Z} \pi $).  Set $x_j=\varphi_{j } -  \varphi_{j\oplus 1}$, then we have $\sum_{j=1}^n x_j=0$ and  $x_j\equiv x_{j\ominus 1} \text { mod } \pi .$ There is only one $x\in [0,\pi)$ such that $x_1 \equiv x \text { mod } \pi$. So we have $$0=\sum_{j=1}^n x_j \equiv \sum_{j=1}^n x=nx \text { mod } \pi.$$
	 	So $x=\frac{k\pi }{n}$  for some $k\in  \{0, 1,2,\cdots, n-1\}.$ And $x_j=\frac{k\pi }{n}+ n_j \pi$ for some $n_j\in \mathbb{Z}.$ Then 
	 	$$c=  \prod_{j=1}^n   \cos  x_j = \prod_{j=1}^n \left[(-1)^{n_j} \cos \frac{k\pi}{n}\right]=(-1)^{e_{\mathbf{n}}} \cos ^n \frac{k \pi }{n}$$
	 	where the exponential $e_{\mathbf{n}}=\sum_{j=1}^n n_j$. For the case $k=0$, i.e., $x=0$, we should have $$e_{\mathbf{n}} \pi= \sum_{j=1}^n n_j\pi=\sum_{j=1}^n x_j=0$$
	 	which implies $e_{\mathbf{n}}=0$ and the corresponding $c=1.$ For all other $k$,
	 	we have $c\geq (-1)^{e_{\mathbf{n}}} \cos ^n \frac{k \pi }{n}\geq -\cos ^n \frac{ \pi }{n}.$  Hence 
	 	$$\mathcal{B}_{n,2}^{\mathbb{R}}\subseteq  [-\cos^n\frac{\pi}{n},1].$$

	 	On the other hand, the value \( -\cos^n \frac{\pi}{n} \) can be achieved by setting \( x_1 = x_2 = \cdots = x_{n-1} = \frac{\pi}{n} \) and \( x_n = -\frac{(n-1)\pi}{n} = -\left( \pi - \frac{\pi}{n} \right) \). This shows that \( [-\cos^n \frac{\pi}{n}, 1] \subseteq \mathcal{B}_{n,2}^{\mathbb{R}} \). Therefore, we conclude that
	 	\[
	 	\mathcal{B}_{n,2}^{\mathbb{R}} = [-\cos^n \frac{\pi}{n}, 1].
	 	\]	  		
 	\end{proof}  
At present, it is not known whether \( \mathcal{B}_{n,2} \cap \mathbb{R} = [-\cos^n \frac{\pi}{n}, 1] \). If this equality holds, the real Bargmann invariant would provide no additional information. However, if the equality does not hold, there would exist some real Bargmann invariant \( r \) from which we could also deduce the imaginarity of the corresponding set.

\vskip 8pt

Note that $\mathcal{B}_{n,d}\subseteq  \mathcal{B}_{n+1,d}.$ In fact, for each $c\in \mathcal{B}_{n,d}$,  there exists  $\Psi=(|\psi_1\rangle,|\psi_2\rangle, \cdots, |\psi_n\rangle)\in \mathbf{P}_2^n $ such that 
$$c=\prod_{j=1}^n \langle \psi_j|\psi_{j\oplus 1}\rangle.$$
And we set $\Psi'=(|\psi_1\rangle,|\psi_2\rangle, \cdots, |\psi_n\rangle,|\psi_n\rangle)\in \mathbf{P}_2^{n+1} $, then 
$$c_{\Psi'}= \left(\prod_{j=1}^{n-1} \langle \psi_j|\psi_{j\oplus 1}\rangle\right) \langle\psi_n|\psi_n\rangle \langle \psi_n|\psi_1\rangle=\prod_{j=1}^n \langle \psi_j|\psi_{j\oplus 1}\rangle.$$ So $c_\Psi= c_{\Psi'} \in \mathcal{B}_{n+1,d}$. It is highly challenging to determine the exact form of the set $\mathcal{B}_{n,d}$. Therefore, it is of particular interest to investigate the limit of these sets, which motivates us to define
\[
\mathcal{B}_{\infty,d} = \bigcup_{n=1}^\infty \mathcal{B}_{n,d}.
\]

We will show that  the limit set $\mathcal{B}_{\infty,d}$  can be fully described by $\mathbb{S}:=\{x+y\mathrm{i}\in \mathbb{C} | x^2+y^2<1, x,y\in \mathbb{R}\}\cup\{1\}.$  The set 
$\mathbb{S}$ represents the open unit disk in the complex plane along with the point  1.

\begin{theorem}\label{thm:ValueDisk}
	Fix an integer $d\geq 2.$ For each $z\in \mathbb{S}$, the value $z$ can be realized in $\mathcal{B}_{n,d}$ for some $n$. That is, there exists some $n\in\mathbb{N}$ and $|\psi_j\rangle\in \mathbf{P}_d$ ($1\leq j\leq n$) such that $z=\tr[\psi_1\psi_2\cdots \psi_n].$ Moreover, we have the identity 
	$$\mathcal{B}_{\infty,2}=\mathcal{B}_{\infty,d} =\mathbb{S}.$$	
\end{theorem}

\begin{proof}
	First, we show that $\mathbb{S}\subseteq \mathcal{B}_{\infty,d}.$ Fix $z\in \mathbb{S}$. If $z=1,$ clearly, $z\in\mathcal{B}_{\infty,d}.$ Otherwise, $z=r e^{\mathrm{i} \theta } $ with $r<1.$ As $\lim\limits_{n\rightarrow \infty } \cos^{n}\frac{\pi}{n}=1$, there exists some $n$ such that $\cos^{n}\frac{\pi}{n}>r.$ Let $\mathcal{D}_{r_n}=\{z_1\in \mathbb{C} \mid |z_1|\leq r_n=\cos \frac{\pi}{n}\} $. By the geometirc intuition of Fig. \ref{fig:Three_points}, we have $\mathcal{D}_{r_n} \subseteq  \mathcal{R}_n$.	Therefore, $$ \mathcal{D}_{r_n^n}= \{z_1^n\mid z_1  \in \mathcal{D}_{r_n} \}\subseteq  \{z_1^n\mid z_1 \in \mathcal{R}_n\} = \mathcal{B}_{n\mid \text{circ}}.$$  As $r< \cos^{n}\frac{\pi}{n}=r_n^n$, we have 
	$$z\in \mathcal{D}_{r_n^n}\subseteq \mathcal{B}_{n\mid \text{circ}} \subseteq \mathcal{B}_{n,2}\subseteq \mathcal{B}_{n,d} .$$ 
	
	Let $c\in \mathcal{B}_{n,d}$. There exists   $\Psi=(|\psi_1\rangle,|\psi_2\rangle, \cdots, |\psi_n\rangle)\in \mathbf{P}_d^n $ such that 
	$c=\prod_{j=1}^n \langle \psi_j|\psi_{j\oplus 1}\rangle.$ Clearly, the modulus $|c|$  of $c$  satisfies
	$|c|\leq 1$. The equality holds if and only if $|\psi_j\rangle =e^{\mathrm{i} \theta_j} |\psi_1\rangle$ whose corresponding $c$ is exactly $1$. For all other cases, $|c|<1$ hence $c\in \mathbb{S}.$ Therefore, $\mathcal{B}_{n,d}\subseteq \mathbb{S}.$
\end{proof}

\section{Conclusions and Discussion }\label{sect:con} 
    
    We have studied the Bargmann invariants, which are pivotal in detecting the quantum imaginarity of sets of quantum states. We have provided complete characterizations of the Bargmann invariant sets \( \mathcal{B}_3 \) and \( \mathcal{B}_{n \mid \text{circ}} \) for all \( n \geq 3 \). For \( \mathcal{B}_n \) (\( n \geq 4 \)), we have examined several of its properties, including multiplicative closure, star-shapedness, and have derived both upper and lower bounds. Furthermore, we have shown that the Bargmann invariants in \( \mathcal{B}_{n \mid \text{circ}} \) can be realized within a qubit system, and have determined the maximum imaginarity for this set. Additionally, we have established that the real Bargmann invariants in \( \mathcal{B}_{n \mid \text{circ}} \) can be realized by real qubit states. Finally, we have completely characterized the set of all Bargmann invariants, denoted by \( \mathcal{B}_{\infty,d} \).  
    
    Our work advances the understanding of Bargmann invariants, emphasizing their critical role in detecting quantum imaginarity in quantum states. Looking ahead, the full characterization of the Bargmann invariant sets \( \mathcal{B}_n \) for \( n \geq 4 \) remains an open challenge. Furthermore, the question of whether there exist real Bargmann invariants that cannot be realized by real states continues to be an intriguing avenue for future research.

\acknowledgements

This work was supported by National Natural Science Foundation of China under Grants No. 12371458, the Guangdong Basic and Applied Basic Research Foundation under Grants Nos. 2023A1515012074, 2024A1515010380  and the Science and Technology Planning Project of Guangzhou under Grants No. 2023A04J1296.

	\appendix
	\section{Proof of Theorem~~\ref{theorem: polygon}}\label{ap:Proof_Thm1}
Before the proof, we make a remark.  First, we only consider  $	G_{\mathbf{z}}=I+z_1 C_n+z_2 C_n^2+\cdots+z_{n-1} C_n^{n-1}$ as a  Hermitian matrix which is equivalent to $\overline{z}_j=\overline{z}_{n-j}$ for all $j=1,2,\cdots,n-1$. As $1,\xi,\xi^2,\cdots,\xi^{n-1}$ are all the eigenvalues of the matrix $C_n$, so the eigenvalues $\mathbf{\lambda}= (\lambda_0,\lambda_1,\cdots,\lambda_{n-1})$ of the matrix $	G_{\mathbf{z}}$ are exactly given by $\lambda_k=\sum\limits_{j=0}^{n-1}z_j\xi^{kj},~k=0,1,...,n-1.$ To simplify the notation, we will always assume $z_0=1,$ and $\mathbf{z}=(z_0,z_1,\cdots,z_{n-1})^T$ as a column vector.  Let $V_n$ denote the set of all $n$-tuples $\mathbf{z}=(z_0,z_1,\cdots,z_{n-1})^T\in \mathbb{C}^n$ with conditions $z_0=1$ and $\overline{z}_j= z_{n-j}$  for $j=1,2,\cdots,n-1$.    Let $F$ denote the $n$ dimensional Fourier transformation, i.e.,   
\begin{equation*}
	F=\begin{bmatrix}
		1 &1&1&\cdots&1 \\
		1 &\xi&\xi^2&\cdots&\xi^{n-1} \\
		1 &\xi^2&\xi^4&\cdots&\xi^{2(n-1)} \\
		\vdots  &\vdots &\vdots &\ddots &\vdots  \\
		1 &\xi^{n-1}&\xi^{2(n-1)}&\cdots &\xi^{(n-1)^2} \\
	\end{bmatrix}.
\end{equation*}	
If $\mathbf{z}=(z_0,z_1,\cdots,z_{n-1})^T\in V_n$, then $\overline{z_j\xi^{kj}}=z_{n-j}\xi^{k(n-j)}$ which implies all eigenvalues  $\lambda_k$ of $	G_{\mathbf{z}}$ are real. Therefore,  the positive semidefinite of $	G_{\mathbf{z}}$ is  equivalent to $\lambda_k\geq 0$ for every $k=0,1,\cdots,n-1$, i.e., $\lambda=F\mathbf{z}\geq \mathbf{0}.$
Then the set $\mathcal{Z}_n$ can be expressed as 
\begin{equation}\label{eq:Z_n_another}
\mathcal{Z}_n=\{z_1\mid F\mathbf{z} \geq \mathbf{0}, \  \mathbf{z}=(z_0,z_1,\cdots,z_{n-1})^T\in V_n \}.
\end{equation}

Now we start to prove the four steps.

{ \bf Step 1: }	 Show that \(1 \in \mathcal{Z}_n\).

Let $\mathbf{z}=(z_0,z_1,\cdots,z_{n-1})^T$ with $z_j=1$ for all $j$. Clearly,   $\mathbf{z}\in V_n$ and  $F\mathbf{z}=(n,0, \cdots, 0)^T\geq \mathbf{0}.$ So $z_1=1\in \mathcal{Z}_n$. 

{ \bf Step 2: }	If \(z_1 \in \mathcal{Z}_n\), then \(\xi z_1 \in \mathcal{Z}_n\) (where \(\xi = e^{\frac{2\pi \mathrm{i}}{n}}\)). 

As $z_1\in \mathcal{Z}_n$, there exists a $\mathbf{z}=(z_0,z_1,\cdots,z_{n-1})^T\in V_n$ whose second entry is exactly $z_1$ such that $F \mathbf{z}\geq \mathbf{0}$.  That is, 
\begin{equation}\label{eq:Fzgeq0}
	\begin{bmatrix}
		1 &1&1&\cdots&1 \\
		1 &\xi&\xi^2&\cdots&\xi^{n-1} \\
		1 &\xi^2&\xi^4&\cdots&\xi^{2(n-1)} \\
		\vdots  &\vdots &\vdots &\ddots &\vdots  \\
		1 &\xi^{n-1}&\xi^{2(n-1)}&\cdots &\xi^{(n-1)^2} \\
	\end{bmatrix}
	\begin{bmatrix}
		1\\
		z_1\\
		z_2\\
		\vdots\\
		z_{n-1}\\
	\end{bmatrix}
	\geq
	\mathbf{0}.
\end{equation}
Set $\mathbf{z}_\xi=(z_0 \xi^0,z_1 \xi^1,\cdots, z_{n-1}\xi^{n-1})$. Clearly, $\overline{z_j\xi^{j}}=z_{n-j}\xi^{n-j}$, for $j=1,2,\cdots, n-1$ and $z_0 \xi^0=1$. Hence $\mathbf{z}_\xi\in V_n$. Moreover, we can check that $F\mathbf{z}_\xi=(C F)\mathbf{z} =C(F\mathbf{z})\geq \mathbf{0}.$  In fact, 
\begin{equation*}
F\mathbf{z}_\xi=	\begin{bmatrix}
		1 &\xi&\xi^2&\cdots&\xi^{n-1} \\
		1 &\xi^2&\xi^4&\cdots&\xi^{2(n-1)} \\
		\vdots  &\vdots &\vdots &\ddots &\vdots  \\
		1 &\xi^{n-1}&\xi^{2(n-1)}&\cdots &\xi^{(n-1)^2} \\
		1 &1&1&\cdots&1 \\
	\end{bmatrix}
	\begin{bmatrix}
		1\\
		z_1\\
		\vdots\\
		z_{n-2}\\
		z_{n-1}\\
	\end{bmatrix}
	\geq
	\mathbf{0}.
\end{equation*}
Therefore, the second entry of $\mathbf{z}_\xi$, i.e., $\xi z_1$, belongs to $\mathcal{Z}_n$.

{ \bf Step 3: }	The set  $\mathcal{Z}_n$ is a convex set.

Given $z_1, z_1'\in \mathcal{Z}_{n}$  and $0\leq p\leq 1$. 
There exist  $\mathbf{z}=(z_0,z_1,\cdots,z_{n-1})^T,\mathbf{z}'=(z_0',z_1',\cdots,z_{n-1}')^T \in V_n$ such that $F\mathbf{z}\geq \mathbf{0}$ and $F\mathbf{z}'\geq \mathbf{0}.$ Set $\mathbf{z}_p:=p\mathbf{z}+(1-p) \mathbf{z}'=(pz_0+(1-p)z_0',pz_1+(1-p)z_1',\cdots,pz_{n-1}+(1-p)z_{n-1}' )$. It is easy to verify that $\mathbf{z}_p \in V_n$  and $$F\mathbf{z}_p=pF\mathbf{z}+(1-p)F\mathbf{z}'\geq \mathbf{0}.$$
Hence, $p z_1+(1-p)z_1'\in \mathcal{Z}_{n} .$

{ \bf Step 4: }  For each $z_1=x+\mathrm{i} y\in \mathcal{Z}_n$,  the point $(x,y)$ must below  the line $\ell_n: x{\cos}\dfrac{\pi}{n}+y{\sin}\dfrac{\pi}{n}-{\cos}\dfrac{\pi}{n}= 0 $ which passing through points $1$ and $\xi$. 

Given   $z_1=x+\mathrm{i} y\in \mathcal{Z}_n,$ $\exists \mathbf{z}=(z_0,z_1,\cdots,z_{n-1})^T\in V_n$   such that $F \mathbf{z}\geq \mathbf{0}$, i.e., Eq.\eqref{eq:Fzgeq0} holds.   In the following, we need to show 
 $$x{\cos}\dfrac{\pi}{n}+y{\sin}\dfrac{\pi}{n}-{\cos}\dfrac{\pi}{n}\leq 0 .$$
 Equivalentlly, 
\begin{equation}\label{eq:ell}
	2{\cos}\dfrac{\pi}{n}-2\left (x{\cos}\dfrac{\pi}{n}+y{\sin}\dfrac{\pi}{n}\right)\geq  0 .
	\end{equation}
Note that $-\left (x{\cos}\dfrac{\pi}{n}+y{\sin}\dfrac{\pi}{n}\right)$ equals to the real part of $z_1(-\cos \dfrac{\pi}{n} + \mathrm{i} \sin \dfrac{\pi}{n} ).$ Set $\mathbf{b}=(b_0, b_1,\cdots,b_{n-1})^T$ where $b_0=2 \cos\dfrac{\pi}{n},$ $b_1=\overline{b}_{n-1} =-\cos \dfrac{\pi}{n} + \mathrm{i} \sin \dfrac{\pi}{n},$ and $b_k=0,  \forall 2\leq k\leq n-2.$ As $\overline{z}_1=z_{n-1}$, Eq. \eqref{eq:ell} can be rewritten as 
\begin{equation}\label{eq:ellz}
	\mathbf{b}^T \mathbf{z} \geq  0 .
\end{equation}
As $F \mathbf{z}\geq \mathbf{0}$, for each $\mathbf{a}=(a_0, a_1,\cdots,a_{n-1})^T\geq \mathbf{0}$, i.e., every entry is nonnegative,  we have 
\begin{equation}\label{eq:ell_aFz}
	\mathbf{a} ^TF \mathbf{z}\geq 0.
\end{equation}
Theofore, to prove Eq. \eqref{eq:ellz},  it is sufficient to find some $\mathbf{a}=(a_0, a_1,\cdots,a_{n-1})^T\geq \mathbf{0}$ such that 
$$\mathbf{a} ^TF=\mathbf{b}^T.$$

In fact,  to ensure the latter condition, $\mathbf{a} ^T$ must be $\mathbf{b}^TF^{-1}$. So we only need to show $\mathbf{a} ^T:=\mathbf{b}^TF^{-1}\geq \mathbf{0}.$ As
\begin{equation*}
	F^{-1}=\dfrac{1}{n}\begin{bmatrix}
		1 &1&1&\cdots&1 \\
		1 &\overline{\xi}&\overline{\xi}^2&\cdots&\overline{\xi}^{n-1} \\
		1 &\overline{\xi}^2&\overline{\xi}^4&\cdots&\overline{\xi}^{2(n-1)} \\
		\vdots  &\vdots &\vdots &\ddots &\vdots  \\
		1 &\overline{\xi}^{n-1}&\overline{\xi}^{2(n-1)}&\cdots &\overline{\xi}^{(n-1)^2} \\
	\end{bmatrix},
\end{equation*}	
 we have 
$$a_k= \dfrac{1}{n}(2\cos \dfrac{\pi}{n}+2 \mathrm{Re}[b_1 \overline{\xi}^k])= \dfrac{2}{n}( \cos \dfrac{\pi}{n} -\cos \dfrac{(2k+1)\pi}{n}).$$
Hence,  $k=0,$ or $n-1$, $a_k=0$ and for   $1\leq k\leq n-2$, $a_k>0.$ This proves   $\mathbf{a} ^T\geq\mathbf{0}.$ 

\qed

  \vskip 5pt
  
  \noindent \emph{Remark:} If $z_1=x+\mathrm{i} y$ lies on the boundary of $\mathcal{Z}_n$ and $\mathbf{z}=(z_0,z_1,\cdots,z_{n-1})^T \in V_n$ such that $F\mathbf{z}\geq \mathbf{0}$, we claim that the rank of corresponding Gram matrix $G_{\mathbf{z}}$ does not greater than 2.   Without loss of generality, we assume $(x,y)$ is on the line $\ell_n:  2{\cos}\dfrac{\pi}{n} -2(x{\cos}\dfrac{\pi}{n}+y{\sin}\dfrac{\pi}{n})= 0 $. Let $\mathbf{a},\mathbf{b}$ be the vectors as in the proof of Theorem \ref{theorem: polygon} and $\mathbf{\lambda}= (\lambda_0,\lambda_1,\cdots,\lambda_{n-1})^T$  be eigenvalues of $G_\mathbf{z}$.
  	Then we have 
  	$$0=2{\cos}\dfrac{\pi}{n} -2(x{\cos}\dfrac{\pi}{n}+y{\sin}\dfrac{\pi}{n})=	\mathbf{a} ^TF \mathbf{z}=\mathbf{a} ^T \lambda.
	$$ 
	At the end of the above proof, we point out that for $k=0,$ or $n-1$, $a_k=0$ and for   $1\leq k\leq n-2$, $a_k>0.$ So we have $a_1\lambda_1+\cdots+a_{n-2}\lambda_{n-2}=0$. From these conditions and $\lambda_k\geq 0$, we must have $\lambda_1=\cdots=\lambda_{n-2}=0.$ Therefore,  $\mathrm{rank}(G_{\mathbf{z}})\leq 2.$   So such Gram matrix $G_{\mathbf{z}}$ can be realized in two dimensional system.

    \bibliography{bibliography}
 
\end{document}

%% file: preamble.tex
\usepackage{amsthm}
\usepackage{mathtools}
\usepackage{physics}
\usepackage{bbm}
\usepackage{amsfonts}
\usepackage{amssymb}
\usepackage{graphicx}

\usepackage[T1]{fontenc}
\usepackage{bbold}
\usepackage{float}
\usepackage[utf8]{inputenc}
\usepackage{csquotes}
\usepackage{tikz-cd}
\usepackage[normalem]{ulem}
\usetikzlibrary{shapes, shapes.geometric, shapes.symbols, shapes.arrows, shapes.multipart, shapes.callouts, shapes.misc,decorations.pathmorphing}
\tikzset{snake it/.style={decorate, decoration=snake}}
\usepackage{tcolorbox}

\usepackage[colorlinks=true,linktocpage=true]{hyperref}
\hypersetup{
    allcolors  = {blue},
}

\theoremstyle{plain}
\newtheorem{theorem}{Theorem}

\newtheorem{proposition}{Proposition}
\newtheorem{lemma}{Lemma}

\definecolor{selectiveyellow}{rgb}{0.8, 0.0, 0.8}

\usepackage{array}
\usepackage{multirow}
\usepackage{xcolor}

\usepackage{tikz}